\documentclass[conference]{IEEEtran}
\IEEEoverridecommandlockouts
\usepackage{cite}
\usepackage{amsmath,amssymb,amsfonts}
\usepackage{algorithmic}
\usepackage{graphicx}
\usepackage{textcomp}
\usepackage{bbm}
\usepackage{xcolor}
\usepackage{balance}
\usepackage[english]{babel}
\usepackage{amsthm}
\usepackage{lipsum}

\newcommand\blfootnote[1]{%
  \begingroup
  \renewcommand\thefootnote{}\footnote{#1}%
  \addtocounter{footnote}{-1}%
  \endgroup
}

\newtheorem{thm}{Theorem}
\newtheorem{claim}{Claim}
\newtheorem{lemma}{Lemma}
\newtheorem{corollary}{Corollary}

\theoremstyle{definition}
\newtheorem{definition}{Definition}
\theoremstyle{definition}
\newtheorem{exmp}{Example}[section]
\theoremstyle{remark}

\newcommand{\remove}[1]{}
\def\BibTeX{{\rm B\kern-.05em{\sc i\kern-.025em b}\kern-.08em
    T\kern-.1667em\lower.7ex\hbox{E}\kern-.125emX}}
\begin{document}

\title{Cascaded Group Testing\\
}

\author{\IEEEauthorblockN{Waqar Mirza}
\IEEEauthorblockA{\textit{Department of EE} \\
\textit{IIT Bombay}\\
Mumbai, India \\
\small{wmirza608@gmail.com}}
\and
\IEEEauthorblockN{Nikhil Karamchandani}
\IEEEauthorblockA{\textit{Department of EE} \\
\textit{IIT Bombay}\\
Mumbai, India \\
\small{nikhilk@ee.iitb.ac.in}}
\and
\IEEEauthorblockN{Niranjan Balachandran}
\IEEEauthorblockA{\textit{Department of Mathematics} \\
\textit{IIT Bombay}\\
Mumbai, India \\
\small{niranj@math.iitb.ac.in}}
}

\maketitle

\begin{abstract}
In this paper, we introduce a variation of the group testing problem where each test is specified by an ordered subset of items and returns the first defective item in the specified order or returns null if there are no defectives. We refer to this as \textit{cascaded group testing} and the goal is to identify a small set of $K$ defective items amongst a collection of size $N$, using as few tests as possible for perfect recovery. For the adaptive testing regime, we show that a simple scheme can find all defective items in at most $K$ tests, which is optimal. For the non-adaptive setting, we first come up with a necessary and sufficient condition for any collection of tests to be feasible for recovering all the defectives. Using this, we show that any feasible non-adaptive strategy requires at least $\Omega(K^2)$ tests. In terms of achievability, it is easy to show the existence of a feasible collection of $O(K^2 \log (N/K))$ tests. We show via carefully constructed explicit designs that one can do significantly better for constant $K$. While the cases $K = 1, 2$ are straightforward, the case $K=3$ is already non-trivial and we come up with an iterative design that is asymptotically optimal and requires $\Theta(\log \log N)$ tests. Note that this is in contrast to standard binary group testing, where at least $\Omega(\log N)$ tests are required. For constant $K \ge 3$, our iterative design requires only $\textup{poly}(\log \log N)$ tests. 
\end{abstract}

\section{Introduction}
\blfootnote{N. Karamchandani's work was supported by a SERB grant on `Online Learning with Constraints' and  a SERB MATRICS grant.}
The problem of \textit{group testing} was originally formulated in \cite{dorfman1943detection}, with the motivation being an efficient screening of syphilis-infected soldiers during World War II. A mathematical model for the problem entails identifying a subset of $K$ defective items from a population of size $N$, using as few \textit{pooled tests} as possible. In standard group testing, each pooled test specifies a subset of the items and produces a binary outcome: `negative' if all the items selected in the test are non-defective, and `positive' otherwise. While the original motivation for group testing was medical testing, it has since found application across a wide variety of scenarios including wireless communications \cite{wolf1985born, sharma2014group, inan2019group}, DNA sequencing \cite{hwang2006pooling}, neighbor discovery \cite{engels2021practical, luo2008neighbor} and network tomography \cite{cheraghchi2012graph,mukamoto2015adaptive}. There is a lot of literature related to the group testing problem and we are only able to discuss a small sample of it below; we refer the interested reader to \cite{hwang2006pooling,du2000combinatorial,malyutov2013search,aldridge2019group} for a more comprehensive coverage of the problem as well as known technical results.

Group testing strategies can be divided into two categories: \textit{adaptive} and \textit{non-adaptive}. The former refers to a setting where the configuration of $i$-th pooled test can be decided based on the outcomes of the previous $i-1$ tests. On the other hand, non-adaptive group testing requires all the tests to be specified beforehand so that they can be potentially conducted in parallel. Since the introduction of the group testing problem, there has been a lot of work on characterizing the optimal number of tests required for identifying $K$ defectives out of $N$ items. A simple counting argument shows that any feasible strategy with $T$ tests should satisfy $2^T \ge {N \choose K}$, and thus $T \ge \log_2 {N \choose K}$ provides a lower bound. There exist adaptive group testing strategies whose performance is very close to this bound; for example, a generalized binary splitting algorithm in \cite{hwang1972method} is guaranteed to require at most $\log_2 {N \choose K} + K$ tests. Similarly, both lower bounds and efficient strategies have been proposed for the non-adaptive setting as well. For example, \cite{d1982bounds, furedi1996onr} showed that any non-adaptive group testing strategy requires at least $\min\{N, \Omega\left(K^2\log_K N\right)\}$ tests. On the other hand, explicit testing strategies based on coding-theoretic ideas were proposed in \cite{kautz1964nonrandom} and later in \cite{porat2008explicit}, which require at most $O\left(K^2\log_K^2 N\right)$ and $O\left(K^2\log N\right)$ tests respectively. 
%

Several variants of the binary `OR' pooled test in standard group testing, as described above, have been studied in the literature.  These include \textit{threshold} testing  \cite{cheraghchi2010improved, chen2009nonadaptive} where the outcome is positive only if the number of defective items included in the test is larger than a specified threshold; \textit{quantitative} testing \cite{soleymani2023non} which has a non-binary outcome like the number of defective items included in the test; and \textit{concomitant} testing \cite{bui2023concomitant}  where a test outcome is positive only if it includes at least one item from each of a pre-defined collection of subsets. 

In this work, we study another testing model with non-binary output, which we call  \textit{cascade testing}. Each test is specified by an ordered subset of items $S = (i_1, i_2, ...,i_r)$, and returns the position of the first defective item in $S$ according to the specified order. Note that this test provides more information than the standard binary `OR' test. As motivation for this model, consider a network tomography application where the goal is to identify the congested links using probes. Each probe traverses through a chosen path in the network and either goes through completely, in which case none of the links on the path are congested; or it returns the identity of the first congested link along the path\footnote{For example, Simple Mail Transfer Protocol (SMTP) can provide such information while selecting efficient paths for email delivery}. Another application would be a movie recommendation system which sequentially presents options from different genres to a user, till he/she picks one; and then uses this feedback to learn the subset of genres that the user likes. Such a feedback model has also recently received significant attention in the Online Learning community under the moniker `cascading bandits' \cite{kveton2015cascading, kveton2015combinatorial, gan2020cost}, with applications in opportunistic spectrum access, network routing, recommendation systems, and dynamic treatment allocation.    

The focus of this work is on characterizing the minimum number of tests required to identify $K$ defectives amongst $N$ items, under the cascade testing model. For the adaptive testing regime, we demonstrate a simple scheme that enables us to find all defective items in at most $K$ tests, which is optimal. For the non-adaptive setting, we first come up with a necessary and sufficient condition for any collection of tests (\textit{design}) to be feasible for recovering all the defectives.  Using this equivalence, we show that any feasible non-adaptive strategy requires at least $\Omega(K^2)$ tests. In terms of achievability, it is easy to show that a collection of $O(K^2 \log N)$ randomly constructed tests is feasible. We show via carefully constructed explicit designs that one can do significantly better for constant $K$. While the cases $K = 1, 2$ are straightforward, the case $K=3$ is already non-trivial and we come up with an iterative design that is asymptotically optimal and requires $\Theta(\log \log N)$ tests. Note that this is in contrast to standard binary group testing, where at least $\Omega(\log N)$ tests are required. For constant $K \ge 3$, our iterative design requires only $\textup{poly}(\log \log N)$ tests.

The rest of the paper is organized as follows. We provide the problem formulation for cascaded group testing in Section~\ref{Sec:PF}. The adaptive and non-adaptive settings are considered in Sections~\ref{Sec:Ad} and \ref{Sec:NonAd} respectively. Bounds on the optimal number of tests needed for non-adaptive group testing are presented in Section~\ref{Sec:Opt}. Finally, some explicit constructions are provided in Section~\ref{Sec:Exp} and a short discussion is presented in Section~\ref{Sec:Disc}. 
\section{Problem Formulation}
\label{Sec:PF}
We have a set of $N$ items, denoted by $V = \{v_1,...,v_N\}$, out of which an unknown subset $\mathcal{K}$ are defective. We assume $|\mathcal{K}| \le K$ for some known $0<K\le N$ and aim to recover $\mathcal{K}$ through `tests'. In standard binary group testing \cite{dorfman1943detection}, each test is specified by a subset $S$ of items and returns 1 if $S$ contains a defective item, i.e., $S\cap \mathcal{K}\neq \emptyset$, and returns 0 otherwise.\\
On the other hand, in this work, we consider an alternate testing model 
 we call \textit{cascaded testing} where a test $t$ is associated with an ordered subset of items $(v_{i_1}, v_{i_2},...,v_{i_{|t|}})$, where $|t|$ denotes the number of items involved in the test. The test returns the first defective item in this sequence. In other words, the result of the test is $y =v_{i_r}$ where\footnote{For $n\in \mathbb{N}$, $[n] :=\{1,2,.., n\}$.} $r=$min$\{j\in [|t|]: v_{i_j} \in \mathcal{K}\}$, assuming there is some defective item in the test. If there are no defective items involved, the test returns\footnote{We impose the constraint $0\notin V$ for consistency.} $y = 0$. Note that a cascaded test provides at least as much information as a standard binary group test.

 The main focus of this work is to characterize the minimum number of tests required in the adaptive and non-adaptive settings under cascaded testing. This involves deriving lower bounds as well as designing achievability schemes. 
 
  We introduce some notation for convenience. We say $x\in t$ (resp. $x\notin t$) to mean $x$ is an item involved in (resp. not in) the test $t$.  
  For $x,y\in t$, we say $x<_t y$ if $x$ appears before $y$ in $t$, and $x\le_t y$ if $x$ appears before $y$ or $x=y$. For any set $U$, let $t \cap U$ ($= U \cap t$) denote the ordered set intersection of $t$ and $U$, with the same order of items as in $t$. Similarly, $t\backslash U:=t\cap U^c$. We will sometimes denote the result of the test as $y= \min(t\cap \mathcal{K})$, where the $\min(.)$ of an ordered set is the item that appears first in the order. The $\min(.)$ is taken as 0 for the empty set. 

  \begin{exmp}
      For $N = 6, K = 2$, consider the test $t = (v_3, v_5, v_2, v_1)$. If $\mathcal{K} = \{v_2, v_5\}$, the result will be $y = \min(t\cap \mathcal{K}) = \min((v_5, v_2)) = v_5$.
  \end{exmp}

\section{Adaptive Setting}
\label{Sec:Ad}
In the adaptive setting, we are allowed to design tests sequentially using the results of previous tests to design the next test. In this setting, we can recover the defective set $\mathcal{K}$ using at most $K$ tests via the following simple procedure.

Initialise $\hat{\mathcal{K}} \leftarrow \emptyset$, $i\leftarrow 1$ and run the following loop:
\begin{enumerate}
    \item Run a test which includes the items in $V\backslash \hat{\mathcal{K}}$ in an arbitrary order.
    \item If the test returns 0 (meaning it found no defectives), terminate the procedure and return $\hat{\mathcal{K}}$. 
    \item If the test returns $v$, then update $\hat{\mathcal{K}} \leftarrow \hat{\mathcal{K}} \cup \{v\}$. 
    \item Update $i\leftarrow i+1$. If $i > K$, terminate the procedure and return $\hat{\mathcal{K}}$.
\end{enumerate}

\remove{
\begin{enumerate}
    \item Initialise $\hat{\mathcal{K}}_1 = \emptyset$. Run the following loop starting from $i=1$.
    \item In the $i$th test include the set $[N]\backslash \hat{\mathcal{K}}_i$ of items in the permutation in arbitrary order.
    \item If the test returns $\infty$ (meaning it found no defectives), terminate the procedure and return $\hat{\mathcal{K}} = \hat{\mathcal{K}}_i$. 
    \item If the test labels item $j \in [N] \backslash \hat{\mathcal{K}}_i$ as defective then update $\hat{\mathcal{K}}_{i+1} = \hat{\mathcal{K}}_i \cup \{j\}$. If $i = K$, terminate the procedure and return $\hat{\mathcal{K}} = \hat{\mathcal{K}}_{K+1}$.
\end{enumerate}
}
The above procedure finds a new defective item in each test (except perhaps the last one if $|\mathcal{K}| < K$), and thus finds all defective items in at most $K$ tests. This is also the minimum number of tests needed (in the worst case of $|\mathcal{K}| = K$) to find all defectives, as each test enables us to find at most a single defective item. Thus, the above simple procedure is in fact optimal. 

\section{Non-Adaptive Setting}
\label{Sec:NonAd}
In the non-adaptive setting, we must fix all the tests beforehand. For a given a collection of tests $\mathcal{T} = \{ t_1, t_2,..., t_T\}$ and a (unknown) defective set $\mathcal{K}$, the outputs of the tests are given by $y_i = \min (t_i\cap \mathcal{K})$ for $i=1,2,..,T$, which we will collectively represent by the output vector $y:= [y_1, y_2,..., y_T]$. Let $\mathcal{X}:= \{S \subset V: |S|\le K\}$ denote the collection of possible defective sets $\mathcal{K}$ and $\mathcal{Y}:=(V\cup\{0\})^T$ denote the set of possible output vectors $y$. Finally, let $\Omega_\mathcal{T}: \mathcal{X} \to \mathcal{Y}$ represent the function that maps every possible defective set to its output vector.

For given $N, K$, we will say that a collection of tests $\mathcal{T}$ is a \textit{feasible testing design} if any set of at most $K$ defectives can be identified using the test outputs. Next, we will identify a necessary and sufficient condition for a collection of tests to be feasible. This can be thought of as an analogue of the test matrix \textit{disjunctness} property that plays the same role under standard binary group testing \cite[Chapter 1]{aldridge2019group}.


\remove{
In a non-adaptive setting, we must fix all the tests beforehand. Suppose we design a deterministic set of $T$ tests given by $\mathcal{T} = \{ t_1, t_2,..., t_T\}$. Suppose $\mathcal{K}$ is some defective set with $|\mathcal{K}| \le K$. Then the  outputs of testing are $y_i =$ min$(t_i(\mathcal{K)})$ for $i=1,2,..,T$. We can collect these into a single output vector $y:= [y_1, y_2,..., y_T]$. Let $\mathcal{X}:= \{S \subset [N]: |S|\le K\}$ and $\mathcal{Y}:=([N]\cup\{\infty\})^T$ and define $\Omega_\mathcal{T} : \mathcal{X} \to \mathcal{Y}$ which maps $\mathcal{K} \mapsto y$. For recovery of every defective set $\mathcal{K}$ from its output vector, we require that $\Omega_\mathcal{T}$ be one-one. We will now prove a necessary and sufficient condition on $\mathcal{T}$ for this to hold.
}

\begin{thm}
\label{Thm:Equivalence}
    Given $N, K$, a testing design $\mathcal{T}$ is feasible if and only if it satisfies the following condition:\\
    For all subsets $\mathcal{K}\subset V$ with $|\mathcal{K}| = K$, and for every $v \in \mathcal{K}$, there exists $t\in \mathcal{T}$ such that: 
    \begin{equation} \label{eqn1}
     v = \min(t\cap \mathcal{K}).
    \end{equation}
\end{thm}
\begin{proof}[Proof Sketch]
In words, the above condition requires that for any possible collection of defectives $\mathcal{K}$ and any item $v \in \mathcal{K}$, there is at least one test $t$ where $v$ appears before every other item in $\mathcal{K}$ and hence $v$ can be identified as defective by $t$. Formally, the proof follows by noting that $\mathcal{T}$ is a feasible testing design if and only if the map $\Omega_\mathcal{T}$ is a one-to-one function, and then proving the necessity and sufficiency of \eqref{eqn1} for this. The detailed proof can be found in Appendix~\ref{Sec:App1}.
\end{proof}

\subsection{Reconstruction Algorithm}
Given the output of a feasible test design satisfying \eqref{eqn1}, one needs a reconstruction algorithm to estimate the set of defectives. The following claim suggests one such strategy. 

\begin{claim}
 For a feasible testing design $\mathcal{T} = \{t_1,...,t_T\}$ and a (unknown) set of defectives $\mathcal{K}$ with $|\mathcal{K}| \le K$, let $y = [y_1,y_2,...y_T] = \Omega_{\mathcal{T}}(\mathcal{K})$ be the output vector.
Then $\mathcal{K}=\{y_i : i \in [T], y_i \neq 0 \}$.
\end{claim}

\begin{proof}
    Follows directly from Theorem \ref{Thm:Equivalence}, which guarantees that each defective $v \in \mathcal{K}$ is detected by at least one test, i.e., $\exists i \in [T]: y_i = v$.
\end{proof}

The above claim provides a natural and simple reconstruction algorithm for cascaded testing. This reconstruction is analogous to the `smallest satisfying set' reconstruction algorithm, which is popular for standard binary group testing \cite[Chapter 2]{aldridge2019group}. 

\section{Optimal testing design}
\label{Sec:Opt}
\begin{definition}[Optimal testing design]
For given $N,K$, a testing design $\mathcal{T}^*$ is said to be optimal if $\mathcal{T}^*$ is feasible and for every feasible design $\mathcal{T}$, $|\mathcal{T}| \ge |\mathcal{T}^*|$ holds. Define:
$$T(N,K):= \text{min}\{|\mathcal{T}|: \mathcal{T} \text{ is feasible} \}$$
\end{definition}
By definition, the size of any optimal design is $T(N, K)$. In the rest of the paper, we derive bounds on the quantity $T(N,K)$. 


To assist with this, we now introduce some notation to write the condition in Theorem~\ref{Thm:Equivalence} more compactly. Let $\mathcal{T}$ be a testing design as before. Given any $S \subset V$ with $S \neq \emptyset$ and some $v \in S$, define:
    \begin{equation}
    \label{Eqn:Ft}
    f_\mathcal{T}(S,v) := |\{t\in \mathcal{T}: v=\text{min}(t \cap S)\}| .
    \end{equation}
This function counts the number of tests in which $v$ appears first among the items in $S$. Using this we can write the condition in Theorem~\ref{Thm:Equivalence} as follows:
\begin{align} \label{eqn2}
    \forall\mathcal{K} \subset [N] \text{ s.t. } |\mathcal{K}| = K,\forall v \in \mathcal{K} : f_\mathcal{T}(\mathcal{K},v) \ge  1 .
\end{align}
We note the following useful property of $f_\mathcal{T}$ for any non-empty $S \subset V$:
\begin{equation} \label{eqn3}
    \sum_{v \in S} f_\mathcal{T}(S,v) \le |\mathcal{T}| .
\end{equation}
This follows from the observation that the sets
$\{t\in \mathcal{T}: v = \text{min}(t\cap S)\}$ are pairwise disjoint for $v\in S$.

We will begin by proving a lower bound on $T(N,K)$, for which we first consider testing designs with a particular `systematic' form. 
%
%
\remove{
To obtain the lower bound, consider an optimal design $\mathcal{T}^*$. From equations \eqref{eqn2} and \eqref{eqn3}, for $S \subset [N]$ with $|S| = K$, we obtain:
$$T(N,K) = |\mathcal{T}^*| \ge \sum_{v\in S} f_\mathcal{T^*}(S,v) \ge \sum_{v \in S} 1 = K$$
}

\subsection{Systematic Form}
\begin{definition}[Systematic form]
    A testing design $\mathcal{T} = \{t_1,...,t_T\}$ is said to be in systematic form if $\forall j\in [T],\ t_j\neq \emptyset$ and min$(t_j) \notin t_k \ \forall k \in [T]\backslash\{j\}$.
\end{definition}
In words, for a testing design in systematic form, the item appearing first in a test is excluded in other tests. Next, we show that to characterize the minimum number of tests required, one can restrict attention to such testing designs.
\begin{thm} \label{thm2}
    Given any feasible testing design $\mathcal{T}_1$, we can construct a feasible testing design $\mathcal{T}_2$ in systematic form satisfying $|\mathcal{T}_2|\le |\mathcal{T}_1|$.
\end{thm}
The proof of the above result can be found in Appendix~\ref{Proof:Thm2}. We present an example to illustrate this result.
\begin{exmp}
    For $N = 4, K = 3$, consider the feasible design $\mathcal{T}_1 = \{(v_1, v_2, v_3, v_4), (v_3, v_2, v_4,v_1),(v_4, v_2, v_3, v_1) \}$. This can be reduced to the systematic form design $\mathcal{T}_2 = \{(v_1, v_2), (v_3, v_2), (v_4, v_2)\}$, which is also feasible.
\end{exmp}
The corollary below immediately follows from Theorem \ref{thm2}.
\begin{corollary} \label{cly1}
For given $N,K$, let $\mathcal{G}_{N,K}$ be the set of feasible testing designs in systematic form. Then
    $$T(N,K) = \text{min}\{|\mathcal{T}|: \mathcal{T} \in \mathcal{G}_{N,K} \} .$$
\end{corollary}

We now present an equivalent condition for the feasibility of systematic form designs.
\begin{thm} \label{thm3}
    Let $\mathcal{T}$ be a testing design in systematic form  and let $L := \{\text{min}(t): t\in \mathcal{T}\}$. Then $\mathcal{T}$ is a feasible testing design if and only if the following holds
    \begin{multline} \label{eqn7}
        \forall S \subset V\backslash L \text{ s.t. } 1\le|S| \le K,\\
        \forall v \in S : f_\mathcal{T}(S,v) \ge  K+1-|S|
    \end{multline} 
\end{thm}

\begin{proof}[Proof Sketch]
Since $\mathcal{T}$ is in systematic form, the first item in each test is distinct, meaning $L$ has exactly $|\mathcal{T}|$ items. To prove the forward implication, we first assume $\mathcal{T}$ is feasible. If \eqref{eqn7} doesn't hold, then there is some set $S \subset V\backslash L$ with size between 1 and $K-1$, and an element $v$ that appears first in $S$ in at most $K-|S|$ tests. Then we can pick a collection $S_1$ of at most $K-|S|$ `blocker items' from $L$, each appearing first in a test where $v$ appears first in $S$. The feasibility condition then fails for the set $S \cup S_1$ where $v$ does not appear first in any test. This proves the forward implication by contradiction.\\
For the converse, we assume \eqref{eqn7} holds. Then we can show that the feasibility condition in Theorem \ref{Thm:Equivalence} holds. For a chosen $\mathcal{K}$, we consider the set of elements $S_1$ of $\mathcal{K}$ that are not in $L$ and the set of elements $S_2$ of $\mathcal{K}$ that are in $L$. Each element of $S_2$ appears first in some test by the definition of $L$. By \eqref{eqn7}, for any element $v$ of $S_1$, there are $K+1-|S_1|=|S_2|+1$ tests where $v$ appears first is $S_1$. In at least one of these tests, $v$ must appear first in $S_1\cup S_2$, since each item in $S_2$ appears in only one test. Thus, the feasibility condition is met for every element in $\mathcal{K}$. The detailed proof can be found in Appendix~\ref{Proof:Thm3}.
\end{proof}

\subsection{Lower Bound}
With Theorem \ref{thm3}, we are now in a position to prove the following lower bound on the number of tests required by any scheme under cascaded testing.
\begin{thm}
\label{Thm:LowerBound}
    For given $N,K$, consider $\alpha,\beta \in \mathbb{N}$ with $\alpha + \beta = K+1$. Then $ N\ge \alpha(\beta+1)-1\  \Rightarrow \ T(N,K)\ge \alpha\beta$.
\end{thm}
\begin{proof}
    By Corollary \ref{cly1}, we have $T(N,K) = $min$\{|\mathcal{T}|:\mathcal{T} \in \mathcal{G}_{N,K}\}$. Consider any $\mathcal{T}\in \mathcal{G}_{N,K}$. Let $T:=|\mathcal{T}|$ and $L:=\{\text{min}(t):t\in \mathcal{T}\}$. Suppose $N\ge \alpha(\beta+1)-1$.\\
    Case 1: If $N-T\ge \alpha \Rightarrow |V|-|L|\ge \alpha$. Thus we can choose $S\subset V\backslash L$ such that $|S| = \alpha$. Since $1\le|S|\le K$, we can apply Theorem \ref{thm3} to obtain $f_\mathcal{T}(S,v) \ge K+1-\alpha = \beta\ \forall v\in S$. Using \eqref{eqn3}, we get:
    $$T \ge \sum_{v\in S}f_\mathcal{T}(S,v)\ge|S|\beta = \alpha \beta$$
    Case 2: If $N-T\le \alpha -1$, we have $T\ge N+1-\alpha \ge \alpha(\beta+1)-\alpha = \alpha \beta$.\\
    
     $\Rightarrow |\mathcal{T}|\ge \alpha\beta\ \forall \mathcal{T}\in \mathcal{G}_{N,K}\Rightarrow T(N,K) \ge \alpha \beta$.
\end{proof}
Choosing $\alpha = \lfloor \frac{K+1}{2}\rfloor$, $\beta = \lceil\frac{K+1}{2}\rceil$ in Theorem~\ref{Thm:LowerBound}, we get the following $\Omega(K^2)$ lower bound on $T(N,K)$.
\begin{corollary}
\label{Cor:LB}
    Suppose $N\ge \lfloor \frac{K+1}{2}\rfloor(\lceil\frac{K+1}{2}\rceil+1)-1$, then:
    $$T(N,K) \ge \left \lfloor \frac{K+1}{2} \right\rfloor \left \lceil\frac{K+1}{2} \right \rceil .$$
\end{corollary}
\remove{
\textcolor{red}{The proof is based on a probabilistic argument an can be found in Appendix~\ref{Sec:RClaim}. The above result implies an upper bound of $O(K^2\log(N / K))$ on $T(N,K)$, which when compared to the $\Omega(K^2)$ lower bound from Corollary~\ref{Cor:LB} shows a gap of at most a factor of $O(\log N)$. Another point to note is that 
As we see next, by carefully designing testing strategies for the cascaded group testing model, one can in fact achieve order-wise fewer tests than the binary testing model.}
}
\section{Feasible testing designs}
\label{Sec:Exp}
In this section, we will present upper bounds on the optimal number of tests $T(N,K)$. Firstly, note that since the cascaded testing model provides at least as much information as the standard binary group testing model, any achievable strategy for the latter is also a feasible one for the former. Thus, an upper bound of $O(K^2\min\{\log^2_K N, \log N\})$ on $T(N,K)$ follows from \cite{kautz1964nonrandom,porat2008explicit} which provided explicit constructions for the binary group testing model. Additionally, an upper bound of $O(K^2\log(N/K))$ can be established via a randomized construction, and is included in Appendix~\ref{Sec:RClaim}. Using the upper bound of $O(K^2\log^2_K N)$ and lower bound of $\Omega(K^2)$, we obtain $T(N,K) = \Theta(K^2)$ in the regime $K=\Theta(N^\alpha)$ for some fixed $\alpha \in (0,1)$. The same asymptotics are true for standard non-adaptive group testing in this regime. As we show next, we obtain a significant improvement in the regime $K = \Theta(1)$.

For any constant $K$ and arbitrary $N$, we present explicit testing designs which use order-wise fewer tests than the upper bounds discussed above. We will first consider $K = 1, 2, 3$ and then describe the general construction. While $K = 1,2$ are straightforward, the case of $K \ge 3$ is much more challenging. For simplicity, we will henceforth consider the set of items $V = \{1,2,..., N\}$ w.l.o.g.
\subsection{$K = 1,2$}
 For $K=1$, we can meet the condition in Theorem~\ref{Thm:Equivalence} using a single test with all items included, i.e., $t_1 =  (1,2,...,N) .$
\\ 
 For $K=2$, we can meet the condition in Theorem~\ref{Thm:Equivalence} using two tests:
    \begin{align*}
        t_1 = (1,2,...,N),\ t_2 = (N,N-1,...,1)
    \end{align*}

The above constructions are also optimal as they achieve the obvious lower bound of $K$ tests. Thus for $K=1,2$ we have $T(N,K)$ independent of the total number of items $N$. Note that this is in contrast to standard non-adaptive binary group testing where for $K = O(1)$,  $\Omega\left(\log N \right)$ tests are necessary. This demonstrates the possible reduction in optimal testing design size that the more informative cascaded testing model might allow for.
\subsection{$K=3$}
For $K=3$, we provide a recursive construction. Suppose we have a feasible testing design $\mathcal{T}_1$ for $n$ items. Additionally, assume that each test in $\mathcal{T}_1$ is a permutation of all $n$ items. We will provide a procedure to construct a feasible design $\mathcal{T}_2$ for $n^2$ items, with $|\mathcal{T}_2| \le |\mathcal{T}_1|+4$.

\noindent {\bf Procedure $\mathcal{A}$:} Let $N=n^2$ and let $V = \{1,...,N\}$ be the set of $N$ items. Partition $V$ into disjoint sets $A_1, A_2, .., A_n$ where $A_i = \{(i-1)n+1, (i-1)n + 2, ..., in \}$.

Next, we introduce some notation which will help us describe the construction. Given permutations $s_1, s_2$ of $n$ items each, we define the permutation $s_3 = s_1 \circ s_2$ of $N = n^2$ items as follows: for each $i$, arrange the $n$ items of $A_i$ according to $s_2$ and call the resulting permutation $h_i$. Next arrange the permutations $h_1, h_2,...,h_n$ according to the order in $s_1$ to obtain the overall permutation $s_3$ on $N$ items. As an example, consider $n = 3$, $N = n^2 = 9$, and $A_i = \{3i-2,3i-1, 3i\}$ for $i = 1,2,3$. If $s_1 = (2,3,1), s_2 = (1,3,2)$, then $h_i = (3i-2, 3i, 3i-1)$ and $s_1 \circ s_2 = (h_2, h_3, h_1) = (4,6,5, 7,9,8, 1,3,2)$.  

Now suppose $\mathcal{T}_1 = \{t_1,..., t_{|\mathcal{T}_1|}\}$. Consider the following set of permutations of $N$ items:

$$\mathcal{F} := \{t_i\circ t_i : i\in [|\mathcal{T}_1|] \}.$$

Now, consider two additional permutations of $n$ items given by $g_1 = (1,2,..., n)$ and $g_2 = (n,n-1,..., 1)$. Let $$\mathcal{H} := \{g_i\circ g_j: i,j\in [2]\}.$$
Note that $|\mathcal{F}|=T$, $|\mathcal{H}| = 4$. Finally, the overall testing design for $N = n^2$ items is given by $\mathcal{T}_2 := \mathcal{F}\cup \mathcal{H}$, and thus $|\mathcal{T}_2| \le |\mathcal{T}_1|+4$.
\begin{exmp}
 For $n = 3$ items $\{1, 2, 3\}$, consider the (trivial) feasible testing design $\mathcal{T}_1= \{(1,2,3), (2, 1, 3), (3, 2, 1)\} = \{t_1,t_2,t_3\}$ with $T = 3$ tests. For $N = n^2 = 9$ items $\{1,..., 9\}$, we use procedure $\mathcal{A}$ to devise a feasible testing design $\mathcal{T}_2$. Set $A_1 = \{1, 2, 3\}$, $A_2 = \{4, 5, 6\}$ and $A_3 = \{7,8,9\}$. Then we can construct $\mathcal{F}$ containing the following permutations: 
 \begin{align*}
\!\!t_1\circ t_1 &\!=\! (1,2,3,4,5,6,7,8,9), \
t_2\circ t_2 \!=\! (5,4,6,2,1,3,8,7,9),\\
&\hspace{.8in} t_3\circ t_3 = (9,8,7,6,5,4,3,2,1).
 \end{align*}
 We can also construct $\mathcal{H}$ using $g_1 = (1,2,3), g_2 = (3,2,1)$ as follows 
\begin{align*}
    \!\!g_1\circ g_1 \!=\! (1, 2, 3, 4, 5, 6, 7, 8, 9), \ 
    g_1 \circ g_2 \!=\! (3,2,1,6,5,4,9,8,7)\\
    \!\!g_2\circ g_1 \!=\! (7,8,9,4,5,6,1,2,3),\
    g_2\circ g_2 \!=\! (9,8,7,6,5,4,3,2,1)
\end{align*}
This gives $\mathcal{T}_2 = \mathcal{F} \cup \mathcal{H}$ as:
\begin{align*}
    \mathcal{T}_2 =\{ (1, 2, 3, 4, 5, 6, 7, 8, 9),
    (9,8,7,6,5,4,3,2,1),\\
    (3,2,1,6,5,4,9,8,7),
    (7,8,9,4,5,6,1,2,3)\\(5,4,6,2,1,3,8,7,9)\}
\end{align*}
It can be verified that $\mathcal{T}_2$ above satisfies the condition in Theorem~\ref{Thm:Equivalence}, and is thus a feasible design for $N = 9, K = 3$. Claim \ref{Claim:Proc} in the next subsection proves that this is true for any testing design matrix constructed using Procedure $\mathcal{A}$.  
\end{exmp}
Since we are guaranteed feasibility of the constructed testing design, we can now recursively apply Procedure $\mathcal{A}$
 starting with $\mathcal{T}_1 = \{(1,2,3),(2,1,3),(3,2,1)\},\ n = 3$, to get designs for any number  of items. Moreover, Procedure $\mathcal{A}$ gives us the following guarantee $\forall n\ge 3$:
 \begin{equation}
 \label{Eqn:recursion}
T(n^2,3) \le T(n, 3) + 4.
\end{equation}
This recursion implies that $T(N, 3) = O(\log\log N)$, which is a significant improvement over standard group testing where at least $\Omega(\log N)$ tests will be required. The following result shows that this construction is asymptotically optimal.

\begin{claim}
\label{erdos-szekeres}
   For $K\ge 3$, we have $T(N,K) > \lfloor \log_2 \log_2 (N-1)\rfloor$.
\end{claim}

\begin{proof}
    The proof uses the Erd\H{o}s-Szekeres(E-S) theorem, which states that for any real sequence of length $n^2+1$ there is a monotone subsequence of length $n+1$. Consider a permutation $t = (u_1,u_2,...,u_{n^2+1})$ of items $\{1,\ldots, n^2+1\}$. By the E-S theorem, there exist $i_1<i_2<...<i_l$ with $l \ge n+1$, such that the sequence of items $u_{i_1}, u_{i_2},...,u_{i_l}$ is monotone. In other words, we can find $n+1$ items that are arranged either in an increasing or decreasing order in $t$.\\
    Let $\mathcal{T} = \{t_1,t_2,...,t_T\}$ be an optimal design for $K=3$ and $N = 2^{2^r}+1$ for arbitrary integer $r\ge 0$. W.l.o.g. we can assume each $t_i$ is a permutation on $[N]$.\\
    Suppose $T\le r$. By the E-S theorem, there exists a subset of $n_1 = 2^{2^{(r-1)}}+1$ items whose relative ordering in $t_1$ is monotone. Now consider the permutation $t_2$ restricted to these items. Again, by E-S, there exists a monotone subsequence of this set of items of size $n_2 = 2^{2^{(r-2)}}+1$. Proceeding inductively, we thus get a subset $A$ of size $n_T = 2^{2^{(r-T)}}+1 \ge 3$, such that each $t_i$ when restricted to $A$ has the items appear in increasing order or decreasing order. But since $n_T\ge 3$, the feasibility condition in Theorem \ref{Thm:Equivalence} is not satisfied and we have a contradiction. Thus, our assumption that $T \le r$ must be incorrect and $T > r = \log _2 \log _2 (N-1)$.\\
    For arbitrary $N \ge 3$, let $r_1 := \lfloor \log_2 \log_2 (N-1)\rfloor$, and $N_1 = 2^{2^{r_1}}+1 \le N$. Then, $T(N,3) \ge T(N_1, 3) > r_1$, that is,
    $$T(N,3) > \lfloor \log_2 \log_2 (N-1)\rfloor.$$
    Since $T(N,K)\ge T(N,3)$ for $K\ge 3$, the claim holds.
\end{proof}

\subsection{$K \ge 3$}
We now consider a generalisation of Procedure $\mathcal{A}$ to any $K\ge 3$. Suppose we have a feasible testing design $\mathcal{T}_1$ for parameters $n,K$ and another feasible design $\mathcal{G}_1$ for parameters $n,(K-1)$. Additionally, assume all tests in these designs are permutations of all $n$ items. We will provide a procedure to construct a feasible design $\mathcal{T}_2$ for $N = n^2$ items, with $|\mathcal{T}_2| \le |\mathcal{T}_1|+|\mathcal{G}_1|^2$.

\noindent {\bf Generalized Procedure $\mathcal{A}$:} Let $N = n^2$ and let $V = \{1,...,N\}$ be the set of $N$ items. Partition $V$ into disjoint sets $A_1, A_2, .., A_n$ where $A_i = \{(i-1)n+1, (i-1)n + 2, ..., in \}$. Now construct $\mathcal{F}, \mathcal{H}$ as follows:
$$\mathcal{F} := \{t\circ t : t\in \mathcal{T}_1\},\ \mathcal{H}=\{g\circ g': g, g'\in \mathcal{G}_1\}.$$
Note that $|\mathcal{F}| = |\mathcal{T}_1|,\ |\mathcal{H}| = |\mathcal{G}_1|^2$. Finally, the overall testing design for $N = n^2$ items is given by $\mathcal{T}_2 := \mathcal{F}\cup \mathcal{H}$, and thus $|\mathcal{T}_2| \le |\mathcal{T}_1|+|\mathcal{G}_1|^2$.

Note that for $K=3$, we can use $\mathcal{G}_1 = \{(1,2,...,n), (n,n-1,...,1)\}$ which is feasible for parameters $n, (K - 1) = 2$. This reduces the procedure to Procedure $\mathcal{A}$ for $K=3$.

The following claim establishes the feasibility of the design obtained from the generalized Procedure $\mathcal{A}$.
\begin{claim}
\label{Claim:Proc}
    The testing design $\mathcal{T}_2$ obtained from the generalized Procedure $\mathcal{A}$ is feasible for parameters $N,K$.
\end{claim}

\begin{proof}[Proof Sketch]
We check that $\mathcal{T}_2$ satisfies the condition for feasibility in Theorem \ref{Thm:Equivalence} for parameters $N, K$. We do this by considering 3 cases for $\mathcal{K} \subset V$: 
\begin{enumerate}
    \item When all items of $\mathcal{K}$ fall into a single partition $A_i$, the feasibility of $\mathcal{T}_1$ enables us to find a permutation in $\mathcal{F}$ for each $v \in \mathcal{K}$, where $v$ is the first defective.
    \item When all items of $\mathcal{K}$ fall into $K$ separate partitions, the idea above still applies and can be used to find suitable permutations in $\mathcal{F}$.
    \item When neither of the above cases happens, there will be $L \le K-1$ partitions into which items of $\mathcal{K}$ fall, and each partition has at most $K-1$ items. Then for each $v \in \mathcal{K}$, the feasibility of $\mathcal{G}_1$ for $n,K-1$ allows us to find a permutation in $\mathcal{H}$, in which $v$ is the first defective.
\end{enumerate} The detailed proof can be found in Appendix~\ref{Sec:AppClaim}.
\end{proof}
 
Given feasible designs for parameters $N,K-1$ for arbitrary $N$, one can use this procedure to recursively construct feasible designs for $N,K$ for arbitrary $N$. Since we have feasible designs for $K = 2$, we can recursively obtain feasible designs for arbitrary $K$.  
This procedure also guarantees that $\forall n\ge K$:
\begin{equation}
\label{recursion}
    T(n^2,K)\le T(n,K)+T(n,K-1)^2.
\end{equation}
This allows us to prove the following result: 
\begin{claim}
\label{Claim:Loglog}
    For a fixed $K \ge 3$, we have $T(N,K) = O((\log \log N )^{a_K})$, where $a_K = 2^{(K-2)}-1$.
\end{claim}
The proof can be found in Appendix~\ref{Sec:Loglog}. This result implies that $T(N,K)$ grows at most as a polynomial in $\log \log N$ for a fixed $K$. This is a significant improvement over standard group testing, where the size of designs is $\Omega(\log N)$.

\section{Discussion}
\label{Sec:Disc}
While we were able to present bounds on the optimal size $T(N,K)$ for cascaded group testing, several pertinent questions remain open. While we were able to demonstrate explicit (near)-optimal constructions for fixed values of $K$, the general case is open. This will also help provide a characterization of the advantage that the additional information in a cascaded test provides over a standard binary OR test. Beyond the zero-error noiseless setting studied here, considering recovery with small error probability and various forms of noise in the test output are interesting directions to pursue.
\remove{
To summarize, we formulated the novel problem of cascaded group testing in adaptive and non-adaptive settings, the former turning out to be trivial. We obtained bounds on the minimum number of tests needed for perfect recovery in the non-adaptive setting. We finally provided explicit feasible designs for the cases $K=1,2,3$.\\
Finding feasible testing designs which are $o(K^2 \log N)$ for general $N, K$, and fully characterizing the asymptotic behaviour of $T(N,K)$ remain open problems. There is much scope in the exploration of other settings within cascaded testing, such as small-error non-adaptive testing, noisy testing, and restricted test designs (such as restricting the number of items or allowing tests only from a predefined set). In the small-error setting, it can easily be seen through a randomized construction that $O(K \log K)$ tests are sufficient to have zero probability of error in the limit. However, the lower bound $K$ tests is order wise smaller and it remains open whether better orders are achievable.
}
\remove{
\section*{Acknowledgment}

The preferred spelling of the word ``acknowledgment'' in America is without 
an ``e'' after the ``g''. Avoid the stilted expression ``one of us (R. B. 
G.) thanks $\ldots$''. Instead, try ``R. B. G. thanks$\ldots$''. Put sponsor 
acknowledgments in the unnumbered footnote on the first page.
}
\balance
\bibliographystyle{IEEEtran}
\bibliography{references}

\newpage

\begin{appendices}
\section{Proof of Theorem~\ref{Thm:Equivalence}}
\label{Sec:App1}
\begin{proof}

 It is easy to see that $\mathcal{T}$ is a feasible testing design if and only if $\Omega_\mathcal{T}$ is a one-to-one function. We will now argue that \eqref{eqn1} is necessary and sufficient for $\Omega_\mathcal{T}$ to be one-to-one.

     First, we prove the forward implication.
     Suppose $\Omega_\mathcal{T}$ is one-one, but condition \eqref{eqn1} does not hold. Then $\exists \mathcal{K} \subset V$ with $|\mathcal{K}| = K$ and $\exists v \in \mathcal{K}$, such that $\forall t \in \mathcal{T}$ one of the following holds: a) $v \notin t$, or b) $v\in t$ and $\exists x \in \mathcal{K} \cap t$ satisfying
    $x <_t v$. \\
    In case a), we have min$(t \cap \mathcal{K})=$min$(t\cap \mathcal{K} \backslash \{v\})$.\\
    In case b), we must have min$(t \cap \mathcal{K}\backslash\{v\})\le_t x <_t v$. Thus, min$(t \cap \mathcal{K} \backslash \{v\}) =$min$(t \cap \mathcal{K})$, since the first item of a test is not changed upon adding an item after it.\\
    Thus $\forall t \in \mathcal{T}$, min$(t \cap \mathcal{K})=$ min$(t \cap \mathcal{K} \backslash \{v\})\Rightarrow \Omega_\mathcal{T}(\mathcal{K}) = \Omega_\mathcal{T}(\mathcal{K} \backslash \{v\})$. This is a contradiction since $\Omega_\mathcal{T}$ was assumed to be one-one.

    Next, we prove the converse, i.e., condition \eqref{eqn1} implies $\Omega_\mathcal{T}$ is one-one. Consider any $S_1, S_2 \in \mathcal{X}$ with $S_1 \neq S_2$. Assume w.l.o.g. that $S_1\backslash S_2 \neq \emptyset$. So $\exists x_1\in S_1$ such that $x_1 \notin S_2$.\\
    Consider some $\mathcal{K} \supset S_1$ with $|\mathcal{K}| = K$. By \eqref{eqn1}, we have that $\exists t \in \mathcal{T}$ such that $x_1 \in t$ and $x_1 <_t x\ \forall \ x \in \mathcal{K} \cap t$ $\Rightarrow$ min$(t \cap S_1) = x_1$. Also, min$(t \cap S_2) \neq x_1$ since $x_1 \notin S_2$.
    Thus, min$(t\cap S_1) \neq$ min$(t \cap S_2)$
    $$ \Rightarrow \Omega_\mathcal{T}(S_1) \neq \Omega_\mathcal{T}(S_2).$$ Since the above argument holds for any distinct $S_1, S_2 \in \mathcal{X}$, we have that 
    \eqref{eqn1} implies $\Omega_\mathcal{T}$ is one-one.
\end{proof}

\section{Proof of Theorem~\ref{thm2}}
\label{Proof:Thm2}
Before we prove Theorem \ref{thm2}, we prove an intermediate result.

\begin{lemma} \label{lemma1}
    Suppose we have a feasible testing design $\mathcal{T} = \{t_1,...,t_T\}$. Given $l\in [T]$ s.t. $t_{l} \neq \emptyset$, define $\hat{t}_i = t_i \backslash \{\text{min}(t_l)\}$ for $i\in [T]\backslash \{l\}$ and $\hat{t}_{l} = t_{l}$. Then the design $\hat{\mathcal{T}} = \{\hat{t}_1,...,\hat{t}_T\}$ is feasible.
\end{lemma}
\begin{proof}
    Let $u:=$min$(t_l)$. We check if $\hat{\mathcal{T}}$ satisfies \eqref{eqn2} for cases when $u \in \mathcal{K}$ and when $u \notin \mathcal{K}$.\\
    If $u \notin \mathcal{K}$, then min$(t\cap \mathcal{K})=$ min$(t\cap \mathcal{K}\backslash\{u\})\ \forall t\in \mathcal{T}$, which implies that $\forall v\in \mathcal{K}$, $f_{\hat{\mathcal{T}}}(\mathcal{K},v) = f_\mathcal{T}(\mathcal{K},v) \ge 1$.\\
    If $u \in \mathcal{K}$, consider $v\in \mathcal{K}$. If $v =u$, we have $f_{\hat{\mathcal{T}}}(\mathcal{K},v)\ge 1$ since $u = \text{min}(\hat{t}_l\cap \mathcal{K})$. If $v\neq u$, consider $j\in [T]$ such that $v = \text{min}(t_j \cap \mathcal{K})$. Then $j\neq l$ and we have min$(\hat{t}_j\cap \mathcal{K})=\text{min}(t_j\cap \mathcal{K}\backslash\{u\}) = v\Rightarrow f_{\hat{\mathcal{T}}}(\mathcal{K},v)\ge 1$.\\
    This covers all cases and thus $\hat{\mathcal{T}}$ is feasible.
\end{proof}

With this lemma, we are ready for the proof of Theorem \ref{thm2}.

\begin{proof}[Proof of Theorem \ref{thm2}]: First, we make a simple observation: Removing an empty test from a feasible design does not change its feasibility.\\
    Applying this idea, we first remove any empty test from $\mathcal{T}_1$ to obtain the feasible design $\{h_1, h_2,..., h_T\}$. Now we apply a procedure to construct a systematic form design.\\
    Initialise $\mathcal{E}\leftarrow \emptyset,\ i \leftarrow 1$ and run the following loop:
    \begin{enumerate}
        \item If $h_i = \emptyset$: Update $\mathcal{E}\leftarrow \mathcal{E} \cup \{i\}$. Update $i\leftarrow i+1$ and start the loop again. 
        \item $\forall j \in [T] \backslash\{i\}$, update $h_j \leftarrow h_j\backslash\{\text{min}(h_i)\}$. 
        \item Update $i\leftarrow i+1$. If $i>T$, terminate the procedure.
    \end{enumerate}
    After the procedure ends, set $\mathcal{T}_2 = \{h_j:j \in[T]\backslash \mathcal{E} \}$.\\
    The design $\{h_1,...,h_T\}$ obtained after each step of the above procedure is feasible by application of Lemma \ref{lemma1}. During the $i$th loop iteration, whenever $h_i\neq \emptyset$, the update in step 2 results in min$(h_i)\notin h_j\ \forall j\in [T]\backslash\{i\}$. Since each iteration can only result in the removal of items from tests, min$(h_i)$ stays excluded from all other tests till the end. Thus in a later iteration $i'$, min$(h_i)$ is not removed from $h_i$, since min$(h_{i'})\neq\text{min}(h_i)$. This guarantees that after the procedure terminates, all empty tests are marked in $\mathcal{E}$, and the remaining tests lead to a systematic form design.
\end{proof}

\section{Proof of Theorem~\ref{thm3}}
\label{Proof:Thm3}
\begin{proof} 
    Observe that the min$(\cdot)$ of distinct tests are distinct in a systematic form design. Define $P:=V\backslash L$ for convenience.\\
    We first prove the forward implication. Suppose $\mathcal{T}$ is feasible and \eqref{eqn7} doesn't hold. Then $\exists S \subset P$ with $1\le|S|\le K$ and $v \in S$ such that $f_\mathcal{T}(S,v) \le K-|S|$.\\
    Let $\mathcal{I} := \{ t\in \mathcal{T}: v = \text{min}(t\cap S)\}$ and note that $|\mathcal{I}| = f_\mathcal{T}(S,v)$. Consider $S_1 := \{\text{min}(t): t\in \mathcal{I}\} \subset L$, with  $|S_1| = |\mathcal{I}|\le K-|S|$. Also, note that $S_1 \cap S = \emptyset$ since $S_1 \subset L$ and $S \subset P$.
    
    Now consider a set of items $S_2\supset S_1$ satisfying $S_2\cap S = \emptyset$ and $|S_2| = K-|S|$. Let $\mathcal{K}:= S \cup S_2 \Rightarrow |\mathcal{K}|= K$. Then $\forall\ t\in \mathcal{I},\ \text{min}(t)\in S_1 \subset \mathcal{K}\Rightarrow \text{min}(t\cap \mathcal{K}) = \text{min}(t) \neq v$. By definition of $\mathcal{I}$, $\forall t\in \mathcal{T} \backslash \mathcal{I},\  \text{min}(t\cap \mathcal{K})\le_t \text{min}(t\cap S) <_t v \Rightarrow\ v\neq \text{min}(t\cap \mathcal{K})$. Thus $\forall t \in \mathcal{T}$, $v \neq \text{min}(t \cap \mathcal{K})$. By Theorem \ref{Thm:Equivalence}, this contradicts our assumption that $\mathcal{T}$ is feasible.\\
    Next, we prove the converse. Suppose \eqref{eqn7} is satisfied. We will check the feasibility of $\mathcal{T}$ using Theorem \ref{Thm:Equivalence}. Consider arbitrary $\mathcal{K} \subset V$ with $|\mathcal{K}| =K$. Let $S_1 = P\cap \mathcal{K}$ and $S_2 = L\cap \mathcal{K}$. Note that $S_1 \cap S_2 = \emptyset$ and $S_1 \cup S_2 = \mathcal{K}$.\\
    For each $v\in S_2 \subset L$, $\exists t\in \mathcal{T}, v = \text{min}(t)$ by the definition of $L$. If $S_1 = \emptyset$, $S_2 = \mathcal{K}$ and we are done.\\
    Suppose $S_1\neq \emptyset$.  For $v \in S_1 \subset P$, we know $f_\mathcal{T}(S_1,v) \ge K+1-|S_1| = |S_2|+1$. Let $\mathcal{J} := \{ t\in \mathcal{T}: v = \text{min}(t\cap S_1)\}$ and let $S_3 := \{\text{min}(t): t\in \mathcal{J}\} \subset L$. Then $|S_3| = |\mathcal{J}| = f_\mathcal{T}(S_1,v) \ge |S_2|+1\Rightarrow S_3\backslash S_2\neq \emptyset$. Consider some $u \in S_3\backslash S_2$ and let $t_u \in \mathcal{J}$ be the test satisfying $u = \text{min}(t_u)$. Since $\mathcal{T}$ is in systematic form, $\forall w\in S_2: w\notin t_u$. So min$(t_u\cap \mathcal{K}) = \text{min}(t_u\cap \mathcal{K}\backslash S_2) = \text{min}(t_u\cap S_1) = v$.\\
    Thus the condition in Theorem \ref{Thm:Equivalence} holds.
\end{proof}

\section{Randomized construction}
We now present an upper bound on $T(N,K)$, which is based on a randomized construction. 
\begin{claim}
\label{Claim:UB}
    $\forall N,K$ with $K\ge 2$,
    $$ T(N,K) \le \left \lfloor \log_{\frac{K}{K-1}} \left( K{N \choose K} \right) \right \rfloor +1 \le \left \lceil K\log \left( K{N \choose K} \right) \right \rceil$$   
\end{claim}

\label{Sec:RClaim}
\begin{proof}
    We use a probabilistic argument to prove this result for any given $N, K$. Choose tests $t_1, t_2,...,t_T$  independently and uniformly at random from the set of permutations of $V$. We look at the probability of the event $E$ that $\mathcal{T} = \{t_1,...,t_T\}$ is not feasible. Using \eqref{eqn2} we obtain
    \begin{equation} \label{eqn6}
        E = \bigcup_{\substack{v \in \mathcal{K} \subset V,\\ |\mathcal{K}|=K}} \{f_\mathcal{T}(\mathcal{K},v) = 0\}.
    \end{equation}
 From \eqref{Eqn:Ft}, we have that for any $\mathcal{K}, v$ in the union, 
 $$
 \{f_\mathcal{T}(\mathcal{K},v) = 0\} = \bigcap_{j\in [T]} A_j^c , \  A_j = \{v \le_{t_j} x, \ \forall x\in \mathcal{K}\} .
 $$
 Since each test $t_j$ is chosen uniformly at random, any element of $\mathcal{K}$ is equally likely to appear first among elements of $\mathcal{K}$. Thus, Pr$[A_j] = 1/K$. From the independence of $t_j$'s, we have
    \begin{equation*}
        \text{Pr}[f_\mathcal{T}(\mathcal{K},v) = 0] = \left(1-\frac{1}{K} \right)^T
    \end{equation*}
    \remove{
    This can be argued as follows. In a given test $t_j$, any element of $\mathcal{K}$ is equally likely to appear first among elements of $\mathcal{K}$. Thus, Pr$[A_j] = 1/K$ where $A_j = \{v \le_{t_j} x, \ \forall x\in \mathcal{K}\}$.\\
    We can write $\{f_\mathcal{T}(\mathcal{K},v) = 0\} = \bigcap_{j\in [T]} A_j^c $ and the result follows from the independence of $t_j$'s.\\}
    
    Applying the union bound to \eqref{eqn6} gives
    \begin{equation*}
        \text{Pr}[E] \le K{N \choose K} \left(1-\frac{1}{K} \right)^T \remove{<  K{N \choose K} e^{-\frac{T}{K}} }
    \end{equation*}
    If $T > \log_{\frac{K}{K-1}}  \left( K{N \choose K} \right)$, we obtain $P[E] < 1$. Thus, there must exist some feasible testing design of size $T =\left \lfloor \log_{\frac{K}{K-1}} \left( K{N \choose K} \right) \right \rfloor +1$, which can be easily upper bounded by $\left \lceil K\log \left( K{N \choose K} \right) \right \rceil$.
\end{proof}

\section{Proof of Claim~\ref{Claim:Proc}}
\label{Sec:AppClaim}
\begin{proof}   

We verify that $\mathcal{T}_2$ satisfies the condition in Theorem \ref{Thm:Equivalence} for parameters $N, K$. Consider some $\mathcal{K} \subset V$, with $|\mathcal{K}| = K$. Then, we have the following cases: 
\begin{enumerate}
    \item Case 1: $\exists i \in [n]:\ \mathcal{K} \subset A_i$, i.e., all $K$ items are in the same partition. By feasibility of $\mathcal{T}_1$, for any $v\in \mathcal{K}$, $\exists t \in \mathcal{T}_1 : v = \min (h_i \cap \mathcal{K})$, where $h_i$ is the items of $A_i$ arranged according to $t$. Let $f := t\circ t \in \mathcal{T}_2$. We have $h_i = (t\circ t) \cap A_i = f \cap A_i$, which gives:
    $$\Rightarrow v = \min (f \cap A_i \cap \mathcal{K}) = \min(f \cap \mathcal{K}).$$
    
    \item Case 2: $|\mathcal{K} \cap A_i| \le 1\ \forall i\in [n]$. This means there are distinct $i_1,i_2,...,i_K \in [n]$, and $v_1, v_2, ..., v_K \in V$ satisfying $v_j \in A_{i_j} \forall j \in [K]$ and $\mathcal{K} = \{v_1,v_2,...,v_K\}$. Let $\mathcal{I} = \{i_1,i_2,...,i_K\}$. Since $\mathcal{T}_1$ is feasible, $\forall j\in [K]$, $\exists t \in \mathcal{T}_1 : i_j=  \min (t\cap \mathcal{I})$. For $k \in [n]$, let $h_k$ denote the items of $A_k$ arranged according to $t$. Let $f = t\circ t \in \mathcal{T}_2$. Since $\mathcal{K} \subset \cup_{i\in \mathcal{I}}A_i$, we have 
    $$\min(f \cap \mathcal{K}) = \min(f \cap(\cup_{i\in \mathcal{I}}A_i) \cap \mathcal{K})$$
    $$ = \min(h_{i_j} \cap \mathcal{K})= v_j$$ 
    The second line above follows from the fact that $i_j$ appears before other elements of $\mathcal{I}$ in $t$, which means $h_{i_j}$ appears before $h_i$ in $f$ for all other $i\in \mathcal{I}$. This combined with the fact that $h_{i_j}$ includes all elements of $A_{i_j}$ implies $h_{i_j} \cap \mathcal{K} = \{v_j\}$, giving us the result.
    
    \item Case 3: $|\mathcal{K} \cap A_i|<K\ \forall i\in [n]$ and $\exists i \in [n]: |\mathcal{K} \cap A_i|>1$. This means there are distinct $A_{i_1},A_{i_2},..., A_{i_L}$ with $1<L\le K-1$, such that $0<|\mathcal{K} \cap A_{i_k}| \le K-1\ \forall k \in [L]$ and $\mathcal{K} \subset \cup _{k\in [L]} A_{i_k}$. For a given $v \in \mathcal{K}$, let $A_{i_j}$ be the partition that contains $v$. Let $\mathcal{K}_j = \mathcal{K} \cap A_{i_j}$, $\mathcal{I} = \{i_k: k\in [L]\}$. Note that $0<|\mathcal{K}_j| \le K-1$. By feasibility of $\mathcal{G}_1$ for parameters $n, K-1$, $\exists g, g' \in \mathcal{G}_1 : i_j = \min(g\cap \mathcal{I})$ and $v = \min(h_{i_j} \cap \mathcal{K}_j)$, where $h_{i_j}$ denotes the items of $A_{i_j}$ arranged according to $g'$. Let $f = g\circ g' \in \mathcal{T}_2$. Then, by the same reasoning as in Case 2, we have
    $$\min (f \cap \mathcal{K}) = \min (f\cap (\cup_{i\in \mathcal{I}}A_i)\cap \mathcal{K})$$
    $$ = \min (h_{i_j} \cap \mathcal{K}) = v$$
\end{enumerate}
In all cases, the condition in Theorem \ref{Thm:Equivalence} holds, which proves the feasibility of $\mathcal{T}_2$.
\end{proof}

\section{Proof of Claim~\ref{Claim:Loglog}}
\label{Sec:Loglog}
\begin{proof}
    We prove this by induction on $K$. First, we show that the claim holds for $K=3$. Firstly, $a_3 = 2^1-1 =1$. Using equation \eqref{Eqn:recursion} and the fact that $T(3,3) = 3$, we get for $r\in \mathbb{N}$ that $T(3^{2^r},3) \le 4r+3$.\\
    For arbitrary $N \ge 3$, let $r := \lceil \log _{2} \log _{3} (N) \rceil$, then $3^{2^r} \ge N$. By monotone property of $T(N,K)$, we have:
    $$T(N,3) \le T(3^{2^r},3) \le 4r+3 = O(\log \log N)$$
    Thus, the claim holds for $K = 3$.\\
    
    Now, suppose the claim holds for some $K\ge 3$. We will show that it also holds for $K+1$. By assumption, there exists $N_1 \in \mathbb{N}, C_1 \in \mathbb{R}$, s.t., $\forall N \ge N_1$:
    $$T(N,K) \le f(N) := C_1 (\log \log N)^{a_K}$$
    We assume $N_1 \ge K+1$ (w.l.o.g.). For arbitrary $r > 0$ consider $N = N_1^{2^r}$. Applying equation \eqref{recursion} repeatedly, we have:
    $$T(N,K+1)\le T(N_1,K+1)+\sum_{l = 0}^{r-1} T(N_1^{2^l},K)^2$$
    $$\le N_1 + r T(N, K)^2\le N_1 + rf(N)^2$$
    We have used the fact that $T(N,K) \le N$ and that $T(N,K)$ is monotone in its parameters. Now $r = \log _{2} \log _{N_1} (N) \le C_2 \log \log N$ for a constant $C_2 > 0$ (assuming $N$ sufficiently large). Thus, we have 
    $$T(N,K+1) \le N_1 + C_1 C_2 (\log \log N )^{2a_K+1}$$
    $$= N_1 +C_1C_2 (\log \log N )^{2^{K-1}-1} \le  C_3(\log \log N)^{a_{K+1}},$$
    for some $C_3 > 0$, assuming $ N \ge N_2$ for a suitable $N_2 \in \mathbb{N}$.\\
    For arbitrary $N\ge N_2$, let $r' := \lceil \log _{2} \log _{N_1} (N) \rceil$, then $N^2 \ge N' := N_1^{2^{r'}} \ge N$. By monotone property of $T(N,K)$, we have $\forall N \ge N_2$:
    $$T(N,K+1) \le T(N',K+1) \le C_3 (\log \log N')^{a_{K+1}}$$
    $$\le C_3(\log \log (N^2))^{a_{K+1}} = O((\log \log N)^{a_{K+1}})$$
    This completes the proof via induction.
\end{proof}

\end{appendices}
\remove{

}
\end{document}